\RequirePackage{fix-cm}
\documentclass[smallextended]{svjour3}    % onecolumn (second format)
\smartqed % flush right qed marks, e.g. at end of proof
\usepackage{graphicx}
\usepackage{algorithm}
\usepackage{algorithmic}
\usepackage{enumitem}
\usepackage{amsmath,amssymb}
\usepackage{tikz}
\usetikzlibrary{arrows.meta}
\usetikzlibrary{matrix}
\usepackage{pgfplots}
\usetikzlibrary{positioning}
\definecolor {processblue}{cmyk}{0.96,0,0,0}
\newtheorem{mydef}{Definition}
\newtheorem{rem}{Remark}
\newtheorem{prop}{Proposition}
\newcommand*\diff{\mathop{}\!\mathrm{d}}
\journalname{Advances in Data Analysis and Classification}
\begin{document}

\title{Independence versus Indetermination: basis of two canonical clustering criteria %\thanks{Grants or other notes
}

%\titlerunning{Short form of title}    % if too long for running head

\author{Pierre Bertrand     \and
    Michel Broniatowski   \and
    \\Jean-François Marcotorchino%etc.
}


\institute{Pierre Bertrand \at
	     Laboratoire de Probabilités, Statistique et Modélisation | CNRS UMR 8001, Sorbonne Université, Paris, France\\
       \email{pierre.bertrand@ens-paris-saclay.fr}      % \\
%       \emph{Present address:} of F. Author % if needed 
      \and
      Michel Broniatowski \at
	     Laboratoire de Probabilités, Statistique et Modélisation | CNRS UMR 8001, Sorbonne Université, Paris, France\\
      \email{michel.broniatowski@sorbonne-universite.fr}
      \and
      Jean-François Marcotorchino \at
	     Institut de statistique, Sorbonne Université, Paris, France\\
	    \email{jfmmarco3@gmail.com}
}

\date{Received: date / Accepted: date}

\maketitle

\begin{abstract}
This paper aims at comparing two coupling approaches as basic layers for building clustering criteria, suited for modularizing and clustering very large networks. 

We briefly use "optimal transport theory" as a starting point, and a way as well, to derive two canonical couplings: "statistical independence" and "logical indetermination". A symmetric list of properties is provided and notably the so called "Monge’s properties", applied to contingency matrices, and justifying the $\otimes$ versus $\oplus$ notation. A study is proposed, highlighting "logical indetermination", because it is, by far, lesser known. 

Eventually we estimate the average difference between both couplings as the key explanation of their usually close results in network clustering.
\end{abstract}

\section{Introduction}
\label{sec:intro}

Network clustering (or cliques partitioning of graphs) is a key topic, concerned with a very large dedicated literature. One of the reasons of this status is the recent and power use made by the GAFAM companies about very large networks resulting of modern activities dealing with: \textit{big social networks, cellphone communications networks, high speed financial trading, large IT networks, IOT networks etc..,}. This is simultaneously associated with the IT capacity afforded today to store the really huge amounts of data, those activities force us to cope with. The sudden apparition of these big networks gave rise to a renewal of the so-called \textit{graph theoretical domain}, used in that context for different purposes, such as: discovering the latent cliques, clustering the whole network, isolating some key parts of interest within the network, etc. In other words, this massive and raw information contained inside the networks must be analyzed \textit{per se}, and this leads obviously to mandatory techniques, among which networks clustering plays a prominent role, with a lot of practical contextual applications. 

In the scientific literature, it appears that many different methods have been dedicated to graphs clustering, one can find in~\cite{fortunato2010community} or more recently in~\cite{doreian2020advances} a quite interesting overview on this matter. Most of them use a local criterion based on the number of paths \cite{katz1953new}, the number of shortest paths~\cite{GN02} or a proportion of present edges~\cite{nascimento2014community} which is then aggregated to define a global criterion to optimize. Some methods are based on pure decomposition of a graph, as for instance in \cite{asano1988clustering} where they construct a $k$-clustering based on a spanning tree by removing $k-1$ edges, some other existing methods are concerned directly with spectral analysis of graph laplacians, or with mathematical relational analysis (correlation clustering). In addition to this relatively main stream list, some methods rely on the application of very specific mathematical domains as those typically addressed through Mean Field Games theory in~\cite{coron2017} or more promising, such as the approach given in~\cite{ni2019community} where a method based on the evolving of a discrete Ricci curvature flow was proposed. 

This method deserves to be briefly exposed: given two measures $\mu$ and $\nu$ in a space endorsed with a distance $d$, the Wasserstein’s transportation distance $W(\mu,\nu)$ is the minimum total weight to move $\mu$ to $\nu$ according to $d$ as presented in~\cite{ollivier2009ricci}. A measure $m^{\alpha,p}_x$ to capture the neighborhood of a node $x$ is formulated in \cite{ni2019community}. With $d$ the shortest path in the network, the Ricci's curvature then basically expresses whether $x$ is closer from $y$ using $d$ or $W$ between $m^{\alpha,p}_x$ and $m^{\alpha,p}_y$. Eventually, they are interested in the Ricci Flow (see~\cite{hamilton1982three}) which solves a differential equation where the derivative is (almost) this Ricci curvature; consequently if the neighborhood between $x$ and $y$ is closer than $d(x,y)$, the derivative is strictly less than one and conversely. Iteratively updating a weight to solve the Ricci Flow differential equation, before cutting links greater than a threshold, they deliver results quite similar and even better on the usual experimental networks. Nevertheless, the lack of a canonical choice for the underlying parameters ($\alpha$,$p$ and the threshold) appears as a limitation similar to the choice of a criterion in the usual Louvain Algorithm. This drawback precisely motivates the present paper. Actually, any comparison with competitive approaches would imply a self adaptive procedure to select the technical parameters. 

At that stage two aspects must be differentiated: on the one hand (i) the existence of generic algorithms to optimize a \textit{clustering criterion} as global objective function, or, on the other hand (ii), on the network clustering criteria themselves. 

Going back on the first point, (i) concerned with generic algorithms, it is well known that several methods were introduced to fit this purpose and notably the famous Louvain algorithm, whose origin is quite recent \cite{Louvain2008}, and which is recognized as a very good tool by the scientific community. It is originally based upon the optimization, of a global function called \textit{modularity} initially defined in \cite{GN04} and which enabled the community to compare two clustering on a common basis. In the sequel we shall quote $\mathcal{M}^{\times}$ this objective function and will show it measures \textit{deviation from statistical independence}.

To fulfill the (ii) objective, the Louvain algorithm has been naturally generalized in \cite{Louvain2013} where the authors proposed to choose a candidate criterion among a list of global criteria, different from the usual \textit{modularity}. Actually, the modularity, due to a resolution limit first mentioned in~\cite{fortunato2007resolution}, has been modified in several articles (\cite{reichardt2006statistical}, \cite{lancichinetti2011limits} or \cite{chen2015new}) always motivated by experimental results; we will not detail further the list of available criteria. In her thesis \cite{PCThese}, Patricia Conde-C{\'e}spedes, proposed some experiments on usual networks, involving $\mathcal{M}^{\times}$ plus some others, showing that results may vary from one criterion to another, while being still consistent and interpretable. 

In this paper, we will focus, on two network clustering criteria she applied, the original $\mathcal{M}^{\times}$ and a second quoted $\mathcal{M}^{+}$ which is locally based on a deviation to another coupling function, already latent in a paper of Fréchet~\cite{FRE51} and that we shall call \textit{indetermination} or \textit{logical indetermination} (notion introduced by J.F. Marcotorchino in his seminal papers~\cite{MAR84} and~\cite{MarcoGSIDual}). 

The innovation of this paper can be stated as follows:
\begin{itemize}
\item We rely on a work of Csiszar on divergences~\cite{csiszar1991least} which assesses that the costs in a projection problem is restricted to Least Square or Entropy. Leveraging on it we show that the two chosen criteria $\mathcal{M}^{\times}$ and $\mathcal{M}^{+}$ precisely result from the optimization of the two corresponding canonical discrete transportation problems.
\item We gather known and new properties of the so-called \textit{indetermination}, an equilibrium already applied in the graph clustering domain but never studied \textit{per se} in a more general context. Furthermore, the expected difference between the two canonical coupling functions is shown to be of order $\mathcal{O}(\frac{1}{n^2})$ where $n$ counts the number of nodes. 
\item We validate some of these findings, by reanalyzing more systematically the behavior of those criteria on the very simple model of Gilbert’s graphs. This last item illustrates small expected difference of the previous item and explains the close experimental results found in~\cite{Louvain2013}. Besides it motivates the search of situations where the two criteria side significantly apart as briefly considered in the paper. 
\end{itemize}

The paper is structured as follows. In section~\ref{sec:1} we propose a parallel discovery of two coupling functions ($\otimes$) and ($\oplus$) using discrete \textit{optimal transport theory}. In section~\ref{sec:Monge} is mentioned a list of dual properties related to Monge's matrices and which justify the notation $\oplus / \otimes$ that we propose. Section~\ref{sec:indetermination} deeply studies \textit{indetermination} introducing properties that, to our knowledge, deserve to be put forward with regards to the too poor coverage which is devoted to them in the literature. Finally, Section~\ref{sec:Graphs} gathers a study about the behavior of the criteria based on those coupling functions on the general Gilbert’s random network model, quoting a global similarity which illustrates their symmetric construction.

\section{Parallel discovery of two dual couplings}
\label{sec:1}
When we want to couple two marginal laws, the most common and straightforward way to proceed, consists in assuming \textit{independence} and keep on computations. For instance when we use a very classical and usual criterion like the $\chi^2$ index, we are measuring nothing but a deviation to \textit{independence} a natural coupling in for empirical experiments.

Although being the most natural, it is not, by far, the only existing available coupling method; actually, as introduced by Sklar in \cite{Sklar73}, any copula function will lead to a coupling function behaving on two cumulative distribution functions. In this document, we link a coupling function to a given optimal transport problem. Hence, to follow a similar approach for \textit{indetermination} coupling, we train ourselves first by extracting \textit{independence} coupling from the optimization of a transport problem and we generalize the principle by applying the same approach to the \textit{indetermination} case, but with a second and different transport problem. 

We already introduced the term "coupling function" several times but let us define it formally, since it will be a key notion throughout the document.

\begin{mydef}[Coupling function]
~\\
\label{def:coupling}
Given $\mu = \mu_1\ldots\mu_p$ and $\nu=\nu_1\ldots\nu_q$ two discrete probabilities called marginal distributions (or simply margins), we want to define a probability function $\pi = \pi_{u,v}~\{1\le u\le p,~1\le v\le q\}$ on the product space.
A way for building it up, consists in making happen a coupling function $C$ such that $\pi = C(\mu,\nu)$, satisfying the following constraints:
\begin{itemize}
 \item (first margin) $ C(\mu,\nu)_{u,\cdot} = \sum_{v=1}^q C(\mu,\nu)_{u,v} = \mu_u,~ \forall 1 \le u \le p $
 \item (second margin) $ C(\mu,\nu)_{\cdot,v} = \sum_{u=1}^p C(\mu,\nu)_{u,v} = \nu_v,~ \forall 1 \le v \le q $
 \item (positivity) $ C(\mu,\nu)_{u,v} \geq 0,~ \forall 1 \le u \le p,~ \forall 1 \le v \le q $
\end{itemize}
\end{mydef}

\begin{rem}
\label{rem:simplecoupling}
~\\
All coupling functions (or maps) we use will satisfy: $\pi_{u,v} = C(\mu,\nu)_{u,v} = C(\mu_u,\nu_v)$; this illustrates that $\pi$ value on $(u,v)$ only depends upon the value on the corresponding margins: $\mu_u$ and $\nu_v$. 
\end{rem}

\subsection{Some few words about Optimal Transport}
\label{ssec:digression}
Looking at Definition~\ref{def:coupling}, we observe that a coupling function behaves as a copula in the discrete domain: acting on margins it derives a probability distribution on the product space.

We can imagine a lot of coupling functions, especially if we do not limit ourselves to Remark~\ref{rem:simplecoupling}. The constraints that $C$ has to respect, lead us to cope with some difficulties. This is the reason why we shall choose a systematic approach: minimizing a cost function and observe the link to \textit{optimal transport} definition.

The ad-hoc discrete optimal transport problem we will be dealing with, typically looks like Problem~\ref{pb:discreteMKP}, given hereafter (where MKP stands for Monge-Kantorovitch-Problem).

\begin{problem}[Discrete Version of MKP]{\label{pb:discreteMKP}}
\begin{eqnarray*}
	\min_{\pi}&&{\sum_{u=1}^p\sum_{v=1}^q{\textbf{C}(\pi(u,v))}}\\
	\textnormal{subject to:} &&\nonumber\\
	&&\sum_{v=1}^{q}{\pi(u,v)} = \mu_u \textnormal{; } \forall u \in \{1,...,p\}\nonumber\\
	&&\sum_{u=1}^{p}{\pi(u,v)} = \nu_v \textnormal{; } \forall v \in \{1,...,q\}\nonumber\\
	&&\pi(u,v) \ge 0 \textnormal{; } \forall (u,v) \in \{1,...,p\}\times\{1,...,q\}\nonumber	
\end{eqnarray*}
\end{problem}

The choice of a cost function $\textbf{C}$ depends upon the applications we want to address. Typically, we expect the global assignment to be as smooth as possible, meaning close to uniform (see both examples in the sequel). A MKP problem is then essentially given by its cost function, while margins $(\mu,\nu)$ may vary. This is the reason why we shall try to solve it with a model taking the fixed margins as parameters. Let us define now an optimal coupling function $C$ associated to a given MKP problem with fixed margins given as parameters.

\begin{mydef}[MKP Problem Associated with Coupling function]
\label{def:couplingOptimal}
~\\
For a given MKP problem $P$, we can define a coupling function $C^P$ by:
$C^P(\mu,\nu) = \pi^*(P)$ provided that $\pi^*$ exists as a unique solution of $P$ with margins $\mu$ and $\nu$.
\end{mydef}

Following Definition \ref{def:couplingOptimal} we propose the solutions of two discrete optimal transport problems that we shall use in section~\ref{sec:Graphs}: each implies a structured and well-defined criterion, suitable for network clustering.

\subsection{The Alan Wilson's Entropy Model: role of "independence"}
\label{ssec:Optimx}

First introduced by Sir Alan Wilson in 1969 for "Spatial Interaction Modeling" the "Flows Entropy Model" of Alan Wilson, can be found in his various publications: originated in \cite{WIL67} and developed in \cite{WIL69}. A fundamental justification of his approach corresponds to the following contextual situation: in a theoretical system, elements of which do not maintain affinities, it is advisable to determine the distribution of $\pi(u,v)$ (normalized frequency flows), supposing $\pi\geq 0$ which maximizes the entropy of the system under certain constraints. The objective function to be minimized is based upon the \textit{Boltzmann's or Shannon's Entropies} so that the problem should be expressed as follows:

\begin{problem}[Unbalanced PSIS]
\label{pb:AWE}
\begin{equation*}
 \min_{\pi}{-\sum_{u=1}^p\sum_{v=1}^q{\pi(u,v)\ln(u,v)}}
\end{equation*}
\end{problem}

In a situation where we have a total absence of information, the minimization of Problem~\ref{pb:AWE} just amounts to satisfy the constraint that the cell values distribution is effectively a probability (i.e.: the sum of positive $\pi(u,v)$ is equal to 1). The solution of this very simple "Program of Spatial Interaction System" (PSIS) is nothing but the uniform law:

\begin{equation}
\pi^*(u,v) = \frac{1}{pq}
\end{equation}

In other words, when we ignore everything about the way the exchanges are built up, it is necessary to use Laplace's principle of "insufficient reason" and to consider that the world trade is uniformly distributed inside the system.

By using margins, let us say information about total exports (origins flows) and total imports (destination flows), degree of disorder of the system can be drastically reduced. Indeed, totals on rows and columns are no longer free, but must satisfy marginal values $\mu_u$ and $\nu_v$, fixed by the application as expressed in Problem~\ref{pb:AWEConstaints}; solution of which is given by theorem~\ref{th:AWE}.

\begin{problem}[Balanced PSIS]\label{pb:AWEConstaints}
\begin{eqnarray*}
 \min_{\pi}{-\sum_{u=1}^p\sum_{v=1}^q{\pi(u,v)\ln(\pi(u,v))}}&&\\
 \textnormal{subject to constraints of Problem~\ref{pb:discreteMKP}}&&
\end{eqnarray*}
\end{problem}

\begin{theorem}\label{th:AWE}
~\\

The solution of Problem~\ref{pb:AWEConstaints} is $\pi^{\times}(u,v) = \mu_u\nu_v $.

Hence the coupling function associated to Problem~\ref{pb:AWEConstaints} is nothing but "independence":
$$C^{Problem~\ref{pb:AWEConstaints}}(\mu,\nu)_{u,v} = C^{\times}(\mu,\nu)_{u,v} = (\mu\otimes\nu)_{u,v} = \mu_u\nu_v$$

\end{theorem}

We skip the proof of theorem~\ref{th:AWE} as it is similar to the one we will develop for theorem~\ref{th:MTM} which is less common. 

As a conclusion, from the direct maximization of entropy, we get the solution expressed in terms of probability and remark that the associated coupling function is nothing but "independence" (expressed with a $\otimes$ throughout the document).

\subsection{The minimal trade model: role of "indetermination"}
\label{ssec:Optim+}

In the "Minimal Trade Model" (see \cite{STE77}, \cite{MAR84} and \cite{MarcoGSIDual}), the cost function aims at getting a smooth breakdown of the origins-destinations $\pi(u,v)= \frac{n_{u,v}}{n_{\cdot,\cdot}}$ which explains the term "Minimal Trade". In that case the criterion is a quadratic function measuring squared deviation of the cells values from the "no information" situation (the uniform joint distribution law related to Problem~\ref{pb:AWE}). Obviously, in case of free margins, the solution remains the uniform law. However, adding usual pre-conditioned constraints on margins, the least squared problem is Problem~\ref{pb:MTM}; solution of which is given by theorem~\ref{th:MTM}.

\begin{problem}[Minimal Trade Model]\label{pb:MTM}
\begin{eqnarray*}
 \min_{\pi}{\sum_{u,v}\left\{\pi(u,v)-\frac{1}{p q}\right\}^2}&&\\
 \textnormal{subject to constraints of Problem~\ref{pb:discreteMKP}}&&
\end{eqnarray*}
\end{problem}

\begin{theorem}\label{th:MTM}
~\\

The solution of Problem~\ref{pb:MTM} is $\pi^{+}(u,v) = \frac{\mu_u}{q} + \frac{\nu_v}{p}-\frac{1}{pq}$.

Hence the coupling function associated to Problem~\ref{pb:MTM} is nothing but "indetermination":
$$C^{Problem~\ref{pb:MTM}}(\mu,\nu)_{u,v} = C^{+}(\mu,\nu)_{u,v} = (\mu\oplus\nu)_{u,v} = \frac{\mu_u}{q} + \frac{\nu_v}{p}-\frac{1}{pq}$$
\end{theorem}

A supplementary condition, which is exogenous with regard to the previous model, can be added on the margins (which are, by the way, constant values given \textit{a priori}), this condition (see \cite{MAR84}) is a simple inequality which guarantees the positivity of the frequency Matrix $\pi^*(u,v)$ we are looking for:

\begin{equation}\label{cond:H}
p \min_{u}{\mu_u} + q \min_v{\nu_v} \ge 1 
\end{equation}

From now on, we shall consider that Condition~\ref{cond:H} applies whatever the breakdown of the $\mu_u$ and $\nu_v$ is. Notice that in the "Adjustment to Fixed Margins for Contingency Table" case, the associated values $n_{uv}$ must be integers, and therefore returns the problem much more complex to solve, relaxation of this integrity constraint leads formally to the Problem~\ref{pb:MTM}.

\begin{rem}[Vanishing bias] 
~\\
By developing the cost function, we obtain an interesting equality we will reuse later on:
\begin{equation}\label{devCarre}
\sum_{u,v}\left(\pi(u,v)-\frac{1}{p q}\right)^2 = \sum_{u,v}\pi^2(u,v) -\frac{1}{p q}
\end{equation}
so that the influence of the constant shift $\frac{1}{pq}$ in the squared model disappears.
\end{rem} 

\begin{proof}
~\\
The proof we propose comes directly from \cite{STE77} and \cite{MarcoGSIDual}. A generalization of the canonic additive form when we relax hypothesis~\ref{cond:H} can be found in the thesis to come~\cite{PBThese}.

Using equality~\ref{devCarre}, the Lagrangian function associated to the previous minimization model can be turned into

\begin{eqnarray*}
\mathbf{L}(\pi,\lambda,\omega,\theta) &=& \sum_{u=1}^p\sum_{v=1}^q{\pi^2(u,v)} - \sum_{u=1}^p \lambda_u\left(\mu_u - \sum_{v=1}^q{\pi(u,v)}\right)\\
&-& \sum_{v=1}^q \omega_v\left(\nu_v - \sum_{u=1}^p{\pi(u,v)}\right) - \theta\left(\sum_{u=1}^p\sum_{v=1}^q{\pi(u,v)}-1\right)
\end{eqnarray*}

Since the function to optimize is a convex one, the solution we are looking for is a minimum so that first order conditions apply and we have the following system of equations.

\begin{eqnarray}
\frac{\partial \mathbf{L}(\pi,\lambda,\omega,\theta) }{\partial\pi(u,v)} &=& 2\pi(u,v)-\lambda_u -\omega_v -\theta = 0\label{MTM1}\\
\frac{\partial \mathbf{L}(\pi,\lambda,\omega,\theta) }{\partial\lambda_u} &=& \mu_u - \sum_{v=1}^q{\pi(u,v)} = 0\label{MTM2}\\
\frac{\partial \mathbf{L}(\pi,\lambda,\omega,\theta) }{\partial\omega_v} &=& \nu_v - \sum_{u=1}^p{\pi(u,v)} = 0\label{MTM3}
\end{eqnarray}

When supposing $\sum_v\omega_v = 0$ as Lagrange multipliers are defined within a constant near we sum \ref{MTM1} on v to obtain $2\mu_u =^{\ref{MTM2}} 2\sum_v \pi(u,v) = q\lambda_u+q\theta$ so that
\begin{equation}
\lambda_u +\theta = \frac{2\mu_u}{q}, \forall u \label{MTM31}
\end{equation} 

From \ref{MTM3} we get $2 \nu_v = \sum_{u=1}^p{2\pi(u,v)} =^{\ref{MTM1}} \sum_{u=1}^p{\lambda_u+\omega_v+\theta} =^{\ref{MTM31}} \sum_{u=1}^p{\frac{2}{q}\mu_u+\omega_v} = \frac{2}{q}\mu_u+p\omega_v$ so that 
\begin{equation}\label{MTM32}
\omega_v = \frac{2\nu_v}{p}-\frac{2}{pq}, \forall v
\end{equation}

Replacing into \ref{MTM1} $\lambda_u+\theta$ and $\omega_v$ by their value given respectively by \ref{MTM31} and \ref{MTM32} we obtain:
\begin{equation*}
\pi^*(u,v) = \frac{\mu_u}{q} + \frac{\nu_v}{p}-\frac{1}{pq}, \forall (u,v)
\end{equation*}

Remark, since Condition~\ref{cond:H} applies, the $\pi^*$ expressed in the previous equation are nonnegative. We will go back to this expression, in the next sections and develop a deeper focus on it, explaining the true meaning of the term "indetermination" and some other consequences.
\end{proof}

\subsection{Expected difference between coupling}
\label{ssec:Ediff}

Both coupling functions are extracted from an optimal transport problem concentrating values around the uniform. Hence differences between them should be small in a certain sense. We provide in this section a measure of their proximity. We evaluate the expected value of a norm between the two couplings under uniform laws. More precisely we suppose the two margins $\mu$ and $\nu$ follow the Dirichlet's law (basically the uniformity on probability distributions). We remind here the form of that law for our application.

\begin{mydef}[Dirichlet's Law]
~\\
The density of a Dirichlet law $\mathcal{D}_p$ representing a uniform law among probability law on $p$ elements is expressed as follows:
$$f(\mu_1,...,\mu_p)\prod_{k=1}^p{\diff \mu_k} = \frac{1}{B(p)} \prod_{k=1}^p\mu_k^{0}\prod_{k=1}^p{\diff \mu_k}= \frac{1}{B(p)} \prod_{k=1}^p{\diff \mu_k}$$ 
where $B$ is the multinomial Beta function.
\end{mydef}

Having expressed a density function for $\mu$ and $\nu$ (replace $p$ by $q$), we apply them two coupling functions $C^{+}$ and $C^{\times}$. As a distance, we define:
\begin{equation*}
	\Delta_p = \mathbb{E}_{(\mu,\nu) \sim \mathcal{D}_p\otimes\mathcal{D}_q} \left[\sum_{u=1}^p\sum_{v=1}^q \left[(\mu\otimes\nu)_{u,v} - (\mu\oplus\nu)_{u,v}\right]^2\right]
\end{equation*}
and compute its value through the sequence:
\begin{eqnarray*}
	\Delta_p &=& \mathbb{E}_{(\mu,\nu) \sim \mathcal{D}_p\otimes\mathcal{D}_q} \left[\sum_{u=1}^p\sum_{v=1}^q \left[(\mu_u-\frac{1}{p})(\nu_v-\frac{1}{q})\right]^2\right]\\
		&=& \mathbb{E}_{\mu \sim \mathcal{D}_p} \left[\sum_{u=1}^p(\mu_u-\frac{1}{p})^2\right] \mathbb{E}_{\nu \sim \mathcal{D}_q} \left[\sum_{v=1}^q(\nu_v-\frac{1}{q})^2\right]\\
		&=& pq\mathbb{E}_{\mu \sim \mathcal{D}_p} \left[(\mu_1-\frac{1}{p})^2\right] \mathbb{E}_{\nu \sim \mathcal{D}_q} \left[(\nu_1-\frac{1}{q})^2\right]
\end{eqnarray*}

Now, we notice that we need to compute the variance of $\mathcal{D}_p$; as it is a known law, we use the following property:

\begin{prop}[Variance of Dirichlet law]
~\\
\label{prop:VarDir}
$\mathbb{V}_{X\sim \mathcal{D}_p}[X] = \frac{p-1}{p^2(p+1)}$
\end{prop}

Proposition~\ref{prop:VarDir} in particular, implies that margins will concentrate their values around $\frac{1}{p}$ and $\frac{1}{q}$ respectively as soon as $p$ or $q$ increases respectively. As we notice that couplings equal each other when any margin is uniform, this should imply that $\Delta_p$ converges to $0$ if any of the two increases. This is exactly what happens, we have the expression:

$$\Delta_p = \frac{1}{pq}\left(\frac{p-1}{p+1}\cdot\frac{q-1}{q+1}\right)\le\frac{1}{pq}$$

\subsection{Structural Justification based upon an axiomatic result of Imre Csiszar}
\label{ssec:justif}

Although it seems arbitrary, our restriction to these two previous coupling functions, is all but a fortuitous decision: in~\cite{csiszar1991least}, Csiszar actually shows that, provided we verify additional intuitive properties, we must restrict ourselves to use either least square or maximum entropy as canonic "distances" between probability distributions.

Let us rewrite our transport problems in terms of the notations he uses in~\cite{csiszar1991least}. We notice that problems~\ref{pb:AWEConstaints} and \ref{pb:MTM} aims at reducing a distance from $\pi$ to the uniform law (that term actually vanishes in both), where $\pi$ must satisfy constraints on its margins leading to an eligible space $L_{\mu,\nu}$ inside the simplex $S_D$, $D>0$. In the first problem, the distance function is the entropy while in the second it is the norm $\mathbb{L}_2$.

A general question is how to adapt a "prior guess" $u^0$ to verify a list of constraints. Let us say $u^0$ lives in $S_D$ while the given constraints define a subspace $L\in\mathcal{L}$ ($\mathcal{L}$ is the space of subspaces of $S_D$ tuned by a finite list of affine constraints, see~\cite{csiszar1991least} for more details).
To formalize it, Csiszar defines a \textit{projection rule} $\Pi$ as a function whose input is a set $L\in\mathcal{L}$ and which generates a method $\Pi_L$ to project any prior guess $u^0$ to a vector in $L$:
\begin{eqnarray*}
\Pi:&& \mathcal{L} \rightarrow \left(S_D \rightarrow S_D\right)\\
&& L \rightarrow \Pi_L : \left(u^0 \rightarrow \Pi_L(u^0) \in L\right) 
\end{eqnarray*}

The article then introduces a collection of "natural" properties that we gather hereafter. 
\begin{itemize}
\item \textit{consistency}: if $L'\subset L$ and $\Pi_L(S_D)\subset L'$ then $\Pi_{L'}=\Pi_L$; basically, if the result of a projection to a bigger space is always inside a smaller, then the projection on the two spaces are equivalent.
\item \textit{distinctness}: if $L$ and $L'$ are defined by a unique constraint and they are not equal, then $\Pi_L\neq \Pi_{L'}$ (unless they both contains the initial prior guess). Typically, in $\mathbb{R}^2$, minimizing $||\cdot||$ on two lines returns a different result as soon as they do not both contain $0$.
\item \textit{continuity}: $\Pi$ is continuous with regards to $L\in\mathcal{L}$; it has a continuous relation with constraints.
\item \textit{scale invariant}: $\Pi_{\lambda L}(\lambda u) = \lambda u$ for any positive $\lambda$ and any $u\in S_D$.
\item \textit{local}: for any subset $J\subset \{1,\ldots,D\}$, $(\Pi_{L})_J = (\Pi_{L'})_{J}$ as soon as $L_J=L'_J$ where $L_J$ means we only keep constraints dealing with coordinates in $J$ and $(\Pi_L)_J$ is the restriction of the resulting vector of $\Pi_L$ to the $J$ coordinates. This property indicates that the results of $\Pi$ on a set of coordinates, only depends on constraints applied to those coordinates. 
\item \textit{transitive}: for any $L'\subset L$, $\Pi_{L'} = \Pi_L'\circ\Pi_{L}$. We can first project on a bigger space without affecting the result.
\end{itemize}

The main result of the paper~\cite{csiszar1991least} states as follows:

\begin{theorem}[Two canonic projections]
~\\
Only two projection rules respect all the conditions quoted beforehand:
\begin{eqnarray*}
\Pi^2 &:& L \rightarrow \Pi^2_L : \left(u^0 \rightarrow \mbox{argmin}_{v\in L} ||v-u^0||_2 \right)\\
\Pi^{KL} &:& L \rightarrow \Pi^{KL}_L : \left(u^0 \rightarrow \mbox{argmin}_{v\in L} \sum_{d=1}^D v_d\ln\left(\frac{v_d}{u^0_d}\right)\right) 
\end{eqnarray*}
where $\Pi^{KL}$ amounts to project using the "Kullback-Leibler" divergence.

\end{theorem}
To come back to our transport problem, the "prior guess" is the uniform law while the subspace $L \subset S_D$ is defined using the margin constraints forced by $\mu$ and $\nu$. Then, provided we verify quoted properties, the two cost functions we used cover an exhaustive view. Eventually it justifies the two graph clustering criteria comparing the neighborhood to each equilibrium (Definition~\ref{def:mx} and Equation~\ref{def:m+}) are canonical.

\section{Monge properties: a justification of the $\oplus/\otimes$ notation}
\label{sec:Monge}

We introduce two classes of matrices, the first one is attributed to Gaspard Monge, from a basic idea appearing in his 1781 paper, (incidentally see\cite{BKR96}, where a reference is given to Alan Hoffman\footnote{In 1961 Alan Hoffman (IBM Fellow and US Science Academy member) rediscovered Monges's observation see~\cite{Hoff63}. Hoffman showed that the Hitchcock–Kantorovich transportation problem can be solved by a very simple approach if its underlying cost matrix satisfies those Monge's properties} who first coined that point and consequently proposed the name: Monge's Matrices). For each of those Monge's matrices, we point out some remarkable equalities and, moreover, we link them to a corresponding coupling function of section~\ref{sec:1}.

\subsection{Monge property -- "Indetermination"}

To introduce Monge’s properties, we follow the exhaustive work of Rainer Burkard, Bettina Klinz and Rüdiger Rudolf exposed in the 66-pages-long article~\cite{BKR96} and begin with definition~\ref{def:Monge}.

\begin{mydef}[Monge and Anti-Monge matrix]\label{def:Monge}
~\\
A $p\times q$ real matrix $c_{u,v}$ is said to be a Monge matrix if it satisfies:
\begin{equation*}
	c_{u,v} + c_{u',v'} \le c_{u',v}+c_{u,v'} ~ \forall~ 1 \le u \le u' \le p, ~1 \le v \le v' \le q
\end{equation*}
and an Anti-Monge matrix if:
\begin{equation*}
	c_{u,v} + c_{u',v'} \geq c_{u',v}+c_{u,v'} ~ \forall~ 1 \le u \le u' \le p, ~1 \le v \le v' \le q
\end{equation*}
\end{mydef}

\begin{rem}[Full-Monge matrix]
~\\
\label{req:FullMonge}
The important case for our purpose is the equality case when a matrix is both Monge and Anti-Monge, we will call this situation "Full-Monge" matrix.
\begin{equation*}
	c_{u,v} + c_{u',v'} = c_{u',v}+c_{u,v'} ~ \forall~ 1 \le u \le u' \le p, ~1 \le v \le v' \le q
\end{equation*}
\end{rem}

Although it is poorly studied, the last introduced equality fits perfectly our purpose. The inequalities on the contrary, are common and can be met in diverse situations such as cumulative distribution functions, or copula theory.

\begin{rem}[Adjacent cells]
~\\
\label{rem:adj}
A straightforward but important derived property is the local adjacency cells equality: it is sufficient to satisfy the property of the remark~\ref{req:FullMonge} on adjacent cells, to ensure the obtainment of a "Full-Monge" matrix behavior for the global set of cells \textit{i.e.}:
\begin{equation*}\label{MongeOne}
	c_{u,v} + c_{u+1,v+1} = c_{u+1,v}+c_{u,v+1} ~ \forall~1 \le u \le p, ~1 \le v \le q
\end{equation*}
\end{rem}

Remark~\ref{rem:adj} is a key property to study Monge matrices since it gives a direct $\mathcal{O}(pq)$ algorithm to verify if a matrix is Monge.

Besides, a question emerges: which density function verifies the full Monge property? The following Proposition~\ref{prop:LabelMonge} gives an interesting answer.

\begin{prop}[Full-Monge matrix is equivalent to "Indetermination"]
~\\
\label{prop:LabelMonge}
A "full Monge matrix" necessarily represents an "indetermination coupling".
\end{prop}

\begin{proof}
~\\
Summing on $u’$ and $v’$ the equality of remark~\ref{req:FullMonge} we straightforwardly obtain:
\begin{equation*}
\sum_{u'}\sum_{v'}\left(c_{u,v}+c_{u',v'}-c_{u',v}-c_{u,v'}\right) = pqc_{u,v}+c_{\cdot,\cdot} -qc_{\cdot,v}-pc_{u,\cdot} = 0 
\end{equation*}
which rewrites:
\begin{equation*}
c_{u,v} = \frac{c_{u,\cdot}}{q} + \frac{c_{\cdot,v}}{p}-\frac{c_{\cdot,\cdot}}{pq}
\end{equation*}

\end{proof}

By summarizing properties we get the following Theorem~\ref{th:Monge}.

\begin{theorem}[Full-Monge matrices]\label{th:Monge}
~\\
$\pi$ representing a probability matrix, the following properties are equivalent.
\begin{enumerate}
\item $\pi$ is a Full-Monge matrix
\item $\pi_{u,v} = \pi^+_{u,v} = \frac{\mu_u}{q} + \frac{\nu_v}{p} - \frac{1}{pq}$
\item $\pi$ optimizes problem~\ref{pb:MTM} for some given margins
\item All $2\times 2$ sub-tables $\{u,v,u',v'\}$ extracted from $\pi$ have the same sum on their diagonal and anti-diagonal
\end{enumerate}
\end{theorem}

Last property of Theorem~\ref{th:Monge} is illustrated on Figure~\ref{fig:ex+} and justifies the $\oplus$ notation assigned to "indetermination". Indeed, if we take blue and red arrows we get the same resulting value: 0. Using the contingency form:
\begin{eqnarray*}
\textnormal{blue arrows} &:& 3+2-1-4 = 0\\
\textnormal{red arrows} &:& 3+2-4-1 = 0
\end{eqnarray*}
Equality obviously remains true for the probability form.

\begin{figure}
\begin{center}
\begin{tabular}{cc}
\begin{tikzpicture}
\matrix (m) [matrix of math nodes]
{
 3 & 4 & 2 & \textbf{9} \\
 2 & 3 & 1 & \textbf{6} \\
 1 & 2 & 0 & \textbf{3} \\
 3 & 4 & 2 & \textbf{9} \\
 \textbf{9} & \textbf{13} & \textbf{5} & \textbf{27} \\
};
\draw[<->,blue] (m-1-1.center) -- (m-3-2.center);
\draw[<->,blue] (m-3-1.center) -- (m-1-2.center);
\draw[<->,red] (m-2-2.center) -- (m-4-3.center);
\draw[<->,red] (m-4-2.center) -- (m-2-3.center);
\end{tikzpicture}
&
\begin{tikzpicture}[thick,scale=0.8, every node/.style={scale=0.8}]
\matrix (m) [matrix of math nodes]
{
 1/9 & 4/27 & 2/27 & \textbf{1/3} \\
 2/27 & 1/9 & 1/27 & \textbf{2/9} \\
 1/27 & 2/27 & 0 & \textbf{1/9} \\
 1/9 & 4/27 & 2/27 & \textbf{1/3} \\
 \textbf{1/3} & \textbf{13/27} & \textbf{5/27} & \textbf{1} \\
};
\draw[<->,blue] (m-1-1.center) -- (m-3-2.center);
\draw[<->,blue] (m-3-1.center) -- (m-1-2.center);
\draw[<->,red] (m-2-2.center) -- (m-4-3.center);
\draw[<->,red] (m-4-2.center) -- (m-2-3.center);
\end{tikzpicture}
\end{tabular}
\caption{Example of an indetermination coupling (Statistical counting vs Probability forms)}
\label{fig:ex+}
\end{center}
\end{figure}
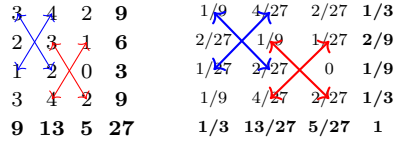

\subsection{Log-Monge property -- Independence}

We present hereafter a similar class of Matrices related, now, to independence: called Log-Monge matrices. They are built on the same principle as before through definition~\ref{def:logMonge}.

\begin{mydef}[Full-Log-Monge Matrices]\label{def:logMonge}~\\
A strictly positive $p\times q$ matrix $c_{u,v}$ is "Full-Log-Monge" when:
\begin{equation*}
	\ln(c_{u,v}) + \ln(c_{u',v'}) = \ln(c_{u',v}) + \ln(c_{u,v'}) ~ \forall~ 1 \le u \le u' \le p, ~1 \le v \le v' \le q
\end{equation*}
\end{mydef}

To immediately get the correspondence, we propose a transposition from a property to another using logarithm in Remark 7. It supposes matrices to be strictly positive (for our probability application: whole discrete space must be reached).

\begin{rem}[From Log-Monge to Monge]\label{req:LogToNormal}
~\\
We easily verify that $c$ satisfies condition proposed in definition~\ref{def:logMonge} if and only if $\ln(c)$ verifies the equivalent condition in definition~\ref{def:Monge} where logarithm is taken element-wise.
\end{rem}

Using Remark~\ref{req:LogToNormal}, we can check that Full-Log-Monge property leads to interesting results and is linked to "independence coupling"; without detailing their obtainment, we gather those results within Theorem~\ref{th:LogMonge}, dual of Theorem~\ref{th:Monge}.

\begin{theorem}[Full-Log-Monge Matrices]\label{th:LogMonge}
~\\
$\pi$ being a strictly positive probability matrix all those properties are equivalent:
\begin{enumerate}
\item $\pi_{u,v}$ is Full-Log-Monge
\item $\pi_{u,v} = \pi^{\times}_{u,v} = \mu_u\nu_v$
\item $\pi$ optimizes problem~\ref{pb:AWE}
\item All $2\times 2$ sub-tables $\{u,v,u',v'\}$ extracted from $\pi$ have the same product on their diagonal and anti-diagonal.
\end{enumerate}
\end{theorem}

Figure~\ref{fig:exx} illustrates "Full Log-Monge" matrices and their properties related to "independence"; it justifies the usual $\otimes$ notation. It is important to remark that both those matrices (in Figure~\ref{fig:ex+} and Figure~\ref{fig:exx}) optimize a problem where the unique difference is the cost functions (since the margins are strictly identical).

\begin{figure}
\begin{center}
\begin{tabular}{cc}
\begin{tikzpicture}[thick,scale=1.1, every node/.style={scale=1.1}]
\matrix (m) [matrix of math nodes]
{
 3 & 13/3 & 5/3 & \textbf{9} \\
 2 & 26/9 & 10/9 & \textbf{6} \\
 1 & 13/9 & 5/9 & \textbf{3} \\
 3 & 13/3 & 5/3 & \textbf{9} \\
 \textbf{9} & \textbf{13} & \textbf{5} & \textbf{27} \\
};
\draw[<->,red] (m-1-1.center) -- (m-3-2.center);
\draw[<->,red] (m-3-1.center) -- (m-1-2.center);
\draw[<->,blue] (m-2-2.center) -- (m-4-3.center);
\draw[<->,blue] (m-4-2.center) -- (m-2-3.center);
\end{tikzpicture}
&
\begin{tikzpicture}
\matrix (m) [matrix of math nodes]
{
 1/9 & 13/81 & 5/81 & \textbf{1/3} \\
 2/27 & 26/243 & 10/243 & \textbf{2/9} \\
 1/27 & 13/243 & 5/243 & \textbf{1/9} \\
 1/9 & 13/81 & 5/81 & \textbf{1/3} \\
 \textbf{1/3} & \textbf{13/27} & \textbf{5/27} & \textbf{1} \\
};
\draw[<->,red] (m-1-1.center) -- (m-3-2.center);
\draw[<->,red] (m-3-1.center) -- (m-1-2.center);
\draw[<->,blue] (m-2-2.center) -- (m-4-3.center);
\draw[<->,blue] (m-4-2.center) -- (m-2-3.center);
\end{tikzpicture}
\end{tabular}
\end{center}
\caption{Example of an "independence coupling" (Contingency vs Probability forms)}
\label{fig:exx}
\end{figure}

\section{Logical "indetermination" and "Condorcet's voting equilibrium"}
\label{sec:indetermination}

In section~\ref{sec:indetermination}, our latent goal is to better understand the "indetermination coupling", that we have until now essentially introduced on a theoretical point of view. Although obtained through a similar process, "independence coupling" is straightforwardly linked to classical empirical experiences. $\pi^+$ does not share this latent simplicity and interpreting it, \textit{per se}, is clearly a domain which deserves to be investigated. We present an attempt for helping the reader to make an accurate picture about the "indetermination" concepts.

Interest for the coupling will be reinforced as we show it corresponds to the Condorcet's majority equilibrium. Defining a "for" vs "against" notion will lead us to a formal equality interpreting "indetermination" in another space. In fact we are faced with the famous "Condorcet's voting equilibrium", which amounts to exhibit the situation where the number of opinions "for" balances exactly the number of opinions "against". The demonstration of this property requires the use of "Mathematical Relational Analysis" notations, which will be formally defined hereafter. We do not want in the context of this article to develop an exhaustive overview of this theory and its applications but pick up some results in connection with the goals we want to achieve; most of them being extracted from the following list of papers which gathers some of the most important key features about the subject: \cite{MAM79}, \cite{MAR84}, \cite{MES90}, \cite{Opitz05}, \cite{MAR86}, \cite{APThese}, \cite{AHP09a}.

We also interpret the equilibrium between the "yes/for" (agreements) and the "no/against" (disagreements) as a voting "indetermination situation". This implies: since the number of votes "for" equals the number of votes "against" we are in a situation, where it is impossible to take a decision. The term: "indetermination" ("indeterminacy" or "uncertainty" should have been used as well) is a formal translation of this surprising situation. First of all, let us introduce properly Relational Analysis notations that we shall use later on.

\begin{mydef}[Relational Analysis notations]
~\\
\label{def:arm_notations}
Let $(u_1,\ldots,u_n)$ and $(v_1,\ldots,v_n)$ be two $n$ probabilistic draws of $U\sim\mu$ and $V\sim\nu$.
We define two associated symmetric $n \times n$ matrices $X$ and $Y$ by 
\begin{eqnarray*}
X_{i,j} = \mathbf{1}_{u_i=u_j} &\textnormal{ and }& Y_{i,j} = \mathbf{1}_{v_i=v_j},~ \forall 1\le i,j \le n
\end{eqnarray*}
\end{mydef}

Basically, the two binary matrices $X$ and $Y$ (which correspond in fact to two binary equivalence relations based on the drawn modalities) represent agreements and disagreements of the two variables on a same draw of size $n$; they are symmetric with $1$ values on their diagonal. This relational coding has a lot of powerful properties, which will not be presented in this paper but which can be found in the articles we mentioned beforehand.

Definition~\ref{def:arm_notations} immediately provides us with an algorithm to transfer contingency representations to relational ones. The way back consists in noticing that:

$X_{i,j} = 1$ if and only if $i$ and $j$ share the same modality of $U\sim\mu$.

Hence we assign a modality to each class defined by the equivalence relation embedded in $X$: the only loss of information during this process resides in the names of modalities.

Now, we are ready to present the Theorem justifying the name "indetermination":

\begin{theorem}[$\pi^+$ and Condorcet's equilibrium]
~\\
\label{th:condorcetEquilibrium}

$\pi$ being a cross probability law on a set of $p\times q$ categorical variables, we shall say that $\pi$ is an "indetermination coupling" on its margins, if and only if the expected number of "agreements" equals the number of "disagreements" on a $2$ independent drawings of $\pi$.
\end{theorem}

\begin{proof}
~\\
Let $\pi$ be a probability law on $p\times q$ categorical variables; it's defined through its values $\pi_{u,v}$, $1\le u\le p$ and $1\le v \le q$. $U$ and $V$ are random variables representing its margins. By $n$ drawings through $\pi$, hence n samplings of $(U,V)$, $U$ and $V$ generates two partitions (equivalence relations) of the $n$ individuals based on their modalities. 

We will say that an agreement occurs when both partitions simultaneously gather or separate the individuals $i$ and $j$. A disagreement occurs on the contrary when a classification regroups $i$ and $j$ while the other one separates them. Formally, if $X,Y$ encodes the $n$ samplings as defined in Definition~\ref{def:arm_notations}:

\begin{itemize}
	\item $X_{i,j}Y_{i,j} = 1$, agreement of type $11$, there are $pq$ couples of classes possible for two individuals $i$ and $j$ to realize this type of agreement
	\item $\overline{X}_{i,j}\overline{Y}_{i,j} = 1$, agreement of type $00$, there are $p(p-1)q(q-1)$ couples of classes of this type
	\item $X_{i,j}\overline{Y}_{i,j} = 1$, disagreement of type $10$, there are $pq(q-1)$ couples of classes of this type
	\item $\overline{X}_{i,j}Y_{i,j} = 1$, disagreement of type $01$, there are $p(p-1)q$ couples of classes of this type
\end{itemize}

As quantities vary according we propose the following equality which establishes that the weighted number of agreements equals the weighted number of disagreements:
\begin{equation}\label{eq:ac_des_eq}
	\frac{XY}{pq} + \frac{\overline{X}\overline{Y}}{p(p-1)q(q-1)} = \frac{X\overline{Y}}{pq(q-1)} + \frac{\overline{X}Y}{p(p-1)q}
\end{equation} 
with the scalar product notation 
$$XY=\sum_{i=1}^n\sum_{j=1}^n\left(X_{i,j}Y_{i,j}\right)$$

Equality~\ref{eq:ac_des_eq} is intrinsically important. It is defined on a draw of size $n$ and linked to a contingency indetermination. 
We take two draws at random independently under $\pi$: $(u_i,v_i)$ and $(u_j,v_j)$ and introduce a probabilistic equality based on our $2$ draws.

\begin{equation}\label{eq:ac_des_eq_prob}
	\frac{\mathbb{E}_{\pi\otimes\pi}\left(X_{i,j}Y_{i,j}\right)}{pq} + \frac{\mathbb{E}_{\pi\otimes\pi}\left(\overline{X}_{i,j}\overline{Y}_{i,j}\right)}{p(p-1)q(q-1)} = \frac{\mathbb{E}_{\pi\otimes\pi}\left(X_{i,j}\overline{Y}_{i,j}\right)}{pq(q-1)} + \frac{\mathbb{E}_{\pi\otimes\pi}\left(\overline{X}_{i,j}Y_{i,j}\right)}{p(p-1)q} 
\end{equation} 

We shall notice now that equality~\ref{eq:ac_des_eq_prob} precisely occurs when $\pi$ equals the indetermination coupling of its margins with the formula introduced in Theorem~\ref{th:MTM}. Let us compute the result of two-sized independent draws under $\pi$.
\begin{itemize}
	\item $\mathbb{E}_{\pi\otimes\pi}(X_{i,j}Y_{i,j}) = \sum_{u_i,v_i}\sum_{u_j,v_j}\pi_{u_i,v_i}\pi_{u_j,v_j}\mathbf{1}_{u_i=u_j \& v_j=v_j} = \sum_{u,v}\pi_{u,v}^2$
	\item $\mathbb{E}_{\pi\otimes\pi}(\overline{X_{i,j}Y_{i,j}}) = \sum_{u_i,v_i}\sum_{u_j,v_j}\pi_{u_i,v_i}\pi_{u_j,v_j}\mathbf{1}_{u_i\neq u_j \& v_i\neq v_j} = \sum_{u,v}\pi_{u,v}(1-\pi_{u,\cdot}-\pi_{\cdot,v} +\pi_{u,v})$
	\item $\mathbb{E}_{\pi\otimes\pi}(X_{i,j}\overline{Y_{i,j}}) = \sum_{u_i,v_i}\sum_{u_j,v_j}\pi_{u_i,v_i}\pi_{u_j,v_j}\mathbf{1}_{u_i=q u_j \& v_i\neq v_j} = \sum_{u,v}\pi_{u,v}(\pi_{u,\cdot}-\pi_{u,v})$ (and similarly for $\mathbb{E}_{\pi\otimes\pi}(\overline{X_{i,j}}Y_{i,j}))$
\end{itemize}

Inserting into equation~\ref{eq:ac_des_eq_prob}, we get:
\begin{eqnarray*}
	&&\frac{\sum_{u,v}\pi_{u,v}^2}{pq} + \frac{\sum_{u,v}\pi_{u,v}(1-\pi_{u,\cdot}-\pi_{\cdot,v} +\pi_{u,v})}{p(p-1)q(q-1)}\\
	&& = \frac{\sum_{u,v}\pi_{u,v}(\pi_{u,\cdot}-\pi_{u,v})}{pq(q-1)} + \frac{\sum_{u,v}\pi_{u,v}(\pi_{\cdot,v}-\pi_{u,v})}{p(p-1)q}\\
	&&\mbox{Reducing to same denominator, we get:}\\
	&&(p-1)(q-1)\sum_{u,v}\pi_{u,v}^2 + \sum_{u,v}\pi_{u,v}(1-\pi_{u,\cdot}-\pi_{\cdot,v} +\pi_{u,v})\\
	&& = (p-1)\sum_{u,v}\pi_{u,v}(\pi_{u,\cdot}-\pi_{u,v}) + (q-1)\sum_{u,v}\pi_{u,v}(\pi_{\cdot,v}-\pi_{u,v})\\
	&&\mbox{regrouping the similar terms yields:}\\
	&& pq\sum_{u,v}\pi_{u,v}^2 -p\sum_{u}\pi_{u,\cdot}^2 -q\sum_{v}\pi_{\cdot,v}^2 +1 =0\\
	&&\mbox{Making use of a classical equality similar to equation~\ref{devCarre}, we obtain:}\\
	&& pq\sum_{u,v}\left(\pi_{u,v} - \pi_{u,\cdot}/q -\pi_{\cdot,v}/p +1/pq\right)^2 =0\\		
	&&\mbox{Finally it holds:}\\
	&& \pi_{u,v} = \frac{\pi_{u,\cdot}}{q} + \frac{\pi_{\cdot,v}}{p} - \frac{1}{pq}
\end{eqnarray*}
We have proved that $\pi=\pi^+$ if and only if the expected number of normalized agreements equals the expected number of disagreements on a $2$-sized drawing.
\end{proof}

In another article~\cite{bertrand:hal-03086553} we interpret the cost function linked with indetermination as a way to reduce couple matchings and show the Condorcet's equilibrium conveys useful applications as it "hides" the underlying distribution.

\section{Application to network clustering}
\label{sec:Graphs}
\subsection{Introduction}

We limit ourselves to Louvain Algorithm, applied as an heuristic to optimize a global objective function. In few words let us say that the global optimization is obtained iteratively by optimizing a local cost function: where two nodes are said to be \textit{similar} if the value of a chosen criterion is high.

Conde Cespedes, in her thesis~\cite{PCThese}, gathered a large amount of network clustering criteria, coming from the scientific literature; she took advantage of this task to give them a category label, depending upon their relationship with both ”independence” or ”indetermination”. She compared them according to their ability to perform on various networks, and collected and stored the obtained results. Although we are in the quite same line with Patricia Conde-C{\'e}spedes, we restrict ourselves to investigate a focused study of both the canonic ones: ”deviation to independence” and ”deviation to indetermination” that we will reintroduce hereafter within the network theoretical context.

First let us start with some usual definitions for a graph:

\begin{mydef}[Weighted graph (network)]
\label{def:graph}
~\\
A weighted graph $G$, is a graph which contains $n$ nodes $1\le i \le n$, which are linked each other through edges $(i,j)$ linked with weights $a_{i,j}$ (representing a weighted incidence matrix). We also introduce the total weight $2M = \sum_{i,j} a_{i,j}$.
\end{mydef}

A basic way to randomly generate a network is through the Gilbert's distribution:
\begin{mydef}[Gilbert]
~\\
Fixing a number $n$ of nodes and $\epsilon \in [0,1]$, we link any set of two nodes by independently drawing though a Bernoulli law with parameter $\epsilon$ leading to a $0-1$ weight. The obtained network is non directed and each weight is $0$ or $1$.
\end{mydef}

\begin{rem}
~\\
Adding a parameter $p$ representing maximum weight, we can easily create a weighted graph by drawing a Binomial law with parameter $(\epsilon,p)$ while linking couples (instead of sets) generates directed networks.
\end{rem}

As mentioned in section~\ref{sec:intro}, our work will be devoted to the research of classes, groupings, clusters or cliques (whatever the name) within a network. They are defined through an equivalence relation as specified in definition~\ref{def:GC}:

\begin{mydef}[Graph clustering]
~\\
\label{def:GC}
Let us call $x$, a matrix representation of a binary equivalence relation, the result of the clustering of a graph $G$. Then $x_{i,j}$ equals $0$ or $1$ and equals $1$ if and only if the two nodes $i$ and $j$ are in the same class for $x$, and $0$ if not.
\end{mydef}

Clustering algorithms aim at providing classes maximizing internal similarities as well as minimizing external ones. A first option is to take as input the number $K$ of classes we are looking for, together with an associated distance (or dissimilarity index) and come up with a list of best representatives or "means" for each class. K-means algorithm whose idea goes back to the fifties (see\footnote{factually this is the method of S. Lloyd(1957) rewritten by E.W. Forgy (1965) which corresponds to the oldest version of the K-means really used} \cite{Steinhaus57}) typically illustrates this option. A second option, is to construct a local criterion $c$ which assigns a weight $c_{i,j}$ to each $(i,j)$ couple of nodes based on their similarity; the more similar they are, the higher the criterion is. We then build up a global criterion by summing up the local values $c_{i,j}$ if and only if $i$ and $j$ are in the same class as proposed in problem~\ref{pb:clustering}.

\begin{problem}[Generic clustering problem]\label{pb:clustering}
\begin{eqnarray*}
 \max_{x}&&{M(c,x) = \sum_{i=1}^n\sum_{j=1}^n{c_{i,j}x_{i,j}}}\\
 \textnormal{subject to:}&&\\
 x && \textnormal{is an equivalence relation (see definition~\ref{def:arm_notations})} 
\end{eqnarray*}
\end{problem}

First let us remark that, as notably spotted in~\cite{MAR86}, \cite{MAM79}, \cite{Opitz05} an equivalence relation constraint can be written as :
\begin{itemize}
\item[$\bullet$] $x_{i,i} = 1, ~ \forall 1\le i \le n$ (reflexivity)
\item[$\bullet$] $x_{i,j} = x_{j,i}, ~ \forall 1\le i,j \le n$ (symmetry)
\item[$\bullet$] $x_{i,j} + x_{j,k} - x_{j,k} \le 1, ~ \forall 1\le i,j,k \le n$ (transitivity)
\end{itemize}

Thanks to the linearity of these constraints, in addition to the linear expression of the criterion itself, the problem~\ref{pb:clustering} although \textit{a priori} NP-hard can be exactly solved (according to some conditions) through the integer relaxation of a good existing 0-1 linear programming code (see~\cite{MAM79}). But in the context of networks clustering, the size $n$ of the problem (here the number of nodes) can be really huge (millions for social networks) and the direct solving by linear programming, even specially tuned, is no longer possible; therefore, the use of robust heuristics becomes mandatory. As mentioned beforehand, Louvain Algorithm (see~\cite{GN02} or \cite{GN04}) belongs to this set of methods: it does not systematically provide us with an exact optimal result but just a quite good approximate one. We concentrate on the analysis of two canonic costs for which historical experiments are reported and explained at the light of the previous sections.

\subsubsection{Original Modularity -- "Independence"}
The original and famous Newman-Girvan's presentation of a global criterion for networks clustering, see~\cite{GN02} or \cite{GN04}, has been introduced in the Louvain algorithm together with a global cost called "Modularity" defined by:

\begin{mydef}[Modularity]
\label{def:mx}
~\\
Given a partition $x_{i,j}$ and a graph $G$ with weighted function $a$ on its edges, the global modularity returns to:
\begin{equation}
M^{\times}(G,x) = \frac{1}{2M}\sum_{i,j} \left[ a_{i,j} - \frac{a_{i,\cdot}a_{\cdot,j}}{2M}\right]x_{i,j}
\end{equation}
\end{mydef}

Let us first remark that the original modularity $M^{\times}$ is nothing but our generic global cost function defined though Problem~\ref{pb:clustering} with:
$$c_{i,j} = m^{\times}(G)_{i,j} = \frac{a_{i,j}}{2M} - \frac{a_{i,\cdot}a_{\cdot,j}}{(2M)^2}$$
and that the local gain $m^{\times}(G)_{i,j}$ to put two nodes in the same class is the local deviation to independence. Indeed, using definition~\ref{def:graph}, with $\pi_{i,j} = \frac{a_{i,j}}{2M}$ as a probability measure on $\{1\ldots n\}^2$ and margins $\mu_i = \frac{a_{i,\cdot}}{2M}$, $m^{\times}$ rewrites:
$$m^{\times}(G)_{i,j} = 2M\left(\pi_{i,j}-\mu_i\mu_j\right)$$
and does express itself as a canonic \textit{deviation to independence} criterion.

A second remark is that as $m^{\times}(G)_{i,j}$ expression does not contain absolute value or square elevation then non connected nodes will lead to negative weights preventing them from being allocated to the same class. If they are connected the importance of $m^{\times}(G)_{i,j}$ evolves positively as $i$ and $j$ have less edges ($a_{i,\cdot}$ and $a_{\cdot,j}$ small); here again this implies an appropriate behavior. More precisely, since independence ensures a coupling as uniform as possible with fixed margins (this is a solution of problem~\ref{pb:AWE}), $m^{\times}$ appears as a fair construction. The criterion basically measures a distance between the observed linkage weight and an expected flat weight given by the average neighborhood.

\subsubsection{Extended Modularity -- "Indetermination"}
We suggest an expression $m^{+}(G)_{i,j}$ which represents a deviation to indetermination. It will be used as a local cost function in Problem~\ref{pb:clustering} leading to a slightly different global formula $M^{+}(G,x)$ to optimize locally:
$$m^{+}(G)_{i,j} = a_{i,j} - \frac{a_{i,\cdot}}{n} - \frac{a_{\cdot,j}}{n} + \frac{2M}{n^2}$$

Symmetrically as $m^{\times}$, it rewrites as a canonic \textit{deviation to indetermination}:
$$m^{+}(G)_{i,j} = 2M*\left(\pi_{i,j}-\frac{\mu_i}{n} - \frac{\mu_j}{n} + \frac{1}{n^2}\right)$$

The global criterion being:
\begin{equation}
\label{def:m+}
M^{+}(G,x) = \sum_{i,j} \left[a_{i,j} - \frac{a_{i,\cdot}}{n} - \frac{a_{\cdot,j}}{n} + \frac{2M}{n^2}\right]x_{i,j}
\end{equation}

We have seen that both couplings share a lot of properties as shown in section~\ref{sec:Monge} and section~\ref{sec:indetermination}. In the same way, Patricia Conde-C{\'e}spedes noticed that a lot of statistical criteria measuring variables correlation are based either on a "distance to independence", or on a "distance to indetermination" (see~\cite{PCThese}). According to these remarks, our canonical \textit{deviation to indetermination} criterion $M^+$ deserves to have the same types of use as those dedicated to the Newman Girvan's $M^{\times}$. 

\subsection{Summary of an application to various networks}
\label{ssec:PatriciaExp}
Now the two criteria are properly introduced with a theoretical basis on their canonical structure, we gather in table~\ref{tab:PCappl} the number of classes found by Patricia Conde-C{\'e}spedes, who applied both on the same empirical networks. She got similar results, as expected beforehand on a bench of experimental graphs she properly defines in~\cite{PCThese} and for which we gather her results in table~\ref{tab:PCappl}. This table can be read as follows: for example, the "Internet" network contains $69,949$ nodes with $351,280$ edges; if we apply Louvain algorithm on, with the global criteria $M^{\times}$ we usually find $46$ communities, while $M^+$ leads to $39$. As anticipated in section~\ref{ssec:Ediff} criteria are (in average) very close (see for instance the Amazon case); consequently their resulting effect on various networks is quite similar. Section~\ref{ssec:Erdos} of the present paper provides the reader with an explanation of the assertion Patricia Conde-C{\'e}spedes experimented.

\begin{table}[!h]
\caption{Number of classes found by each criteria on various networks}
\label{tab:PCappl} 
\begin{tabular}{lllllll}
\hline\noalign{\smallskip}
& Karate & Football & Jazz & Internet & Amazon & YouTube\\
N (nb nodes) & 34 & 115 & 198 & 69 949 & 334 863 & 1 134 890\\
M (sum of weights) & 78 & 613 & 2 742 & 351 280 & 925 872 & 2 987 624\\
\noalign{\smallskip}\hline\noalign{\smallskip}
Number of classes for criteria $M^{\times}$ & 4 & 10 & 4 & 46 & 250 & 5 567\\
Number of classes for criteria $M^+$ & 4 & 10 & 6 & 39 & 246 & 13 985\\
\noalign{\smallskip}\hline
\end{tabular}
\end{table}

\subsection{Gilbert Experimental Tests}
\label{ssec:Erdos}

As already mentioned, solving problem~\ref{pb:clustering} is NP-hard so that we cannot expect precise results, neither about the number of classes for a given criterion, nor about their composition. Nevertheless, we can compare directly local values of the criterion to extrapolate a common or a distinct global run when iteratively optimized.

We propose a comparative try based on Gilbert's networks to spot differences or similarity between $m^{\times}(G)_{i,j}$ and $m^{+}(G)_{i,j}$ values. The aim is to observe the distribution of both criteria on a typical network. First, to simplify observations and as only the reference cost (\textit{i.e.} one equilibrium of section~\ref{sec:1}) varies between $m^+$ and $m^{\times}$, we only keep it by subtracting $a_{i,j}$; it is formally defined in definition~\ref{def:bias}. Then, we generate $1,000,000$ networks randomly, compute each criterion on a random pairs of nodes and store the reference cost; the results are gathered within figure~\ref{fig:critScatter}.

\begin{mydef}[Bias or reference cost]
~\\
\label{def:bias} 
The two bias derived from $m^{\times}$ and $m^{+}$ are respectively:
\begin{eqnarray*}
b^{\times}_{i,j} = \frac{a_{i,\cdot}a_{\cdot,j}}{2M} & \textnormal{ and } & b^{+}_{i,j} = \frac{a_{i,\cdot}}{n} + \frac{a_{\cdot,j}}{n} - \frac{2M}{n^2}
\end{eqnarray*}
\end{mydef}

\begin{figure}[h]
\begin{center}
\includegraphics[scale=0.5]{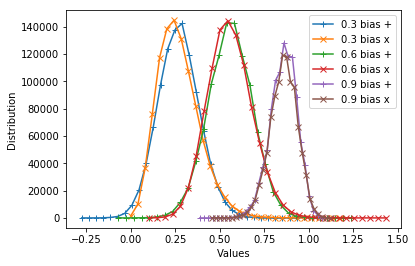}
\caption{Empirical distribution of the two reference costs $b^+_{i,j}$ and $b^{\times}_{i,j}$ for $\epsilon$ in [0.3, 0.6, 0.9]; X-axis gives the values of the bias, Y-axis gives the corresponding number of realizations} 
\label{fig:critScatter}
\end{center}
\end{figure}

On figure~\ref{fig:critScatter} we observe that the distributions of both biases are similar for any values of $\epsilon$. Indeed, the curves are identical on their core values (those with a number of realizations upon $20,000$). It illustrates their common origin which amounts to flatten a distribution (section~\ref{sec:1}) and leads to a small expected difference (section~\ref{ssec:Ediff}). Their common mean is equal to $\epsilon$ as it can be easily derived from the formulas. A difference nevertheless remains on extreme values particularly visible for $\epsilon=0.3$. Let us now compute theoretically both distributions under Gilbert's networks to confirm their symmetry.

\begin{prop}[Probability values]
~\\
\label{prop:probValues}
Let $b$ be a binary value, $b \le n_i \le n$ and $b \le n_j \le n$; let us compute the following probability:
\begin{eqnarray*}
&&\mathbb{P}(a_{i,j} = b, a_{i,\cdot} = n_i, a_{\cdot,i} = n_j)\\
&& = \epsilon^b(1-\epsilon)^{1-b} \binom{n-1}{n_i-b}\epsilon^{n_i-b}(1-\epsilon)^{n-1-n_i+b} \binom{n-1}{n_j-b}\epsilon^{n_j-b}(1-\epsilon)^{n-1-n_j+b}
\end{eqnarray*}
\end{prop}
The corresponding value $m^+_{i,j}$ and $m^{\times}_{i,j}$ associated to a group $(b,n_i,n_j)$ of the parameters being evident, we propose figure~\ref{fig:critDistTh} which represents the difference between theoretical distributions of both criteria with $\epsilon=0.3$. $b^{\times}$ and $b^{+}$ have distinct forms but their proximity on highly probable values, given on Figure~\ref{fig:critDistTh}, illustrates section~\ref{ssec:Ediff}: if we couple two variables with $n$ margins, expected difference is less than $\frac{1}{n^2}$. 

\begin{figure}[!h]
\begin{center}
\includegraphics[scale=0.5]{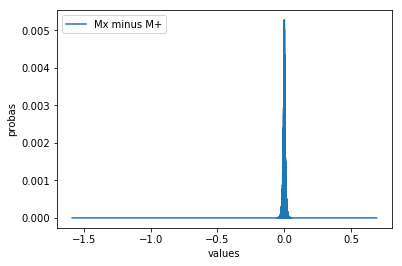}
\caption{Theoretical distribution of the difference $m^{\times}(G)_{i,j} - m^{+}(G)_{i,j}$ (same as $b^+_{i,j}-b^{\times}_{i,j})$ on generated graphs}
\label{fig:critDistTh}
\end{center}
\end{figure}

Extreme values, on the contrary may differ drastically. Though it seems the opposite to Figure~\ref{fig:critScatter} as $m^+$ comes with higher values than $m^\times$, it's consistent because of the \textit{minus} sign in the formula linking $m$ with $b$. Having noticed that $b^+$ and $b^{\times}$ differ on their extreme values, we compute them on a general Gilbert network (respecting the common value of $2M=n^2\epsilon$), and obtain the bounds:
\begin{eqnarray}
-\epsilon \le  b^+ \le \frac{n}{n} + \frac{n}{n} -\epsilon = 2-\epsilon  &\textnormal{ and }& 0 \le b^{\times} \le \frac{n\times n}{n^2\epsilon} = \frac{1}{\epsilon} \label{eq:extremebias}
\end{eqnarray}

As already expected with figure~\ref{fig:critDistTh} the difference between extreme values is arbitrarily high. Eventually it shows that in average, both canonical criteria will share a similar behavior hence shall be applied indifferently when the network's neighborhood weight $a$ is close to uniform. Precisely, in any real application the chosen criterion distinguishes at each iteration a link among the neighborhood unexpected when compared to the equilibrium (either independance or indetermination). The difference between the two equilibria is so small that if one distinguishes a link, the other will so that we expect the resulting classes to be comparable if not identical. Indeed, we know we can build, following equations~\ref{eq:extremebias}, networks on which $\mathcal{M}^+$ and $\mathcal{M}^{\times}$ deliver completely different results.

\section{Conclusions}
First, we followed the historical line and introduced two basis from Discrete Optimal Transport Theory: \textit{independence} and \textit{indetermination}. As recalled, the first one is the most intuitive and frequently used in mathematical articles as well as experimented in real life. The second notion appeared more surprising, poorly studied in the statistical literature but more commonly used by people working on Mathematical Relational Analysis Voting Theory and Analysis of Variance.
Together, they cover the only two canonic projection costs as quoted in \cite{csiszar1991least}.

To illustrate the usefulness of the parallel construction, we turned to applications and completed the track followed by Patricia Conde-C{\'e}spedes in her thesis~\cite{PCThese}. She gathered a list of networks clustering criteria and classified them according to their deviation to one of the mentioned coupling functions. Section~\ref{sec:Graphs} reports a further analyze of the two canonical criteria. It gathers results about the general similarity of their application on various networks as well as their extreme values to set one another apart. 

In each section, from optimal transport to networks, we insisted on the parallel between both notions together with their differences. As quoted beforehand, they appear as the two unique canonic structural solutions. Generally, the differences between them needs to be scanned up, either to coin a macro criteria, or to chose wisely between one or another depending on the structure of the network. In any case, the traditional use of \textit{independence} at the expense of \textit{indetermination} needs to be further investigated and explained.

%\label{sec:2}
%as required. Don't forget to give each section
%and subsection a unique label (see Sect.~\ref{sec:1}).
%\paragraph{Paragraph headings} Use paragraph headings as needed.
%\begin{equation}
%\end{equation}

%\begin{figure}
%% Use the relevant command to insert your figure file.
%% For example, with the graphicx package use
% \includegraphics{example.eps}
%\label{fig:1}  % Give a unique label
%\end{figure}
%\begin{figure*}
%% Use the relevant command to insert your figure file.
%% For example, with the graphicx package use
% \includegraphics[width=0.75\textwidth]{example.eps}
%\label{fig:2}  % Give a unique label
%\end{figure*}
%\begin{table}
%\label{tab:1}  % Give a unique label
%\begin{tabular}{lll}
%\hline\noalign{\smallskip}
%first & second & third \\
%\end{tabular}
%\end{table}

\begin{acknowledgements}
We thank the editor and two anonymous referees for their valuable comments.
\end{acknowledgements}


%\bibliographystyle{spbasic}  % basic style, author-year citations
\bibliographystyle{spmpsci}  % mathematics and physical sciences
\bibliography{Biblio} % name your BibTeX data base

\end{document}